\documentclass[12pt]{iopart}
\usepackage{iopams}  
\expandafter\let\csname equation*\endcsname\relax
\expandafter\let\csname endequation*\endcsname\relax
\usepackage{amsmath}
\usepackage{amsfonts}
\usepackage{amscd}
\usepackage{amsthm}
\usepackage{setspace}
\usepackage{amssymb}
\usepackage{braket}
\usepackage{hyperref}
\usepackage{color}

\newtheorem{theorem}{Theorem}
\begin{document}

\title[]{Rigorous convergence condition for quantum annealing}

\author{Yusuke Kimura$^1$ and Hidetoshi Nishimori$^1$ $^2$ $^3$}

\address{$^1$ International Research Frontiers Initiative, Tokyo Institute of Technology, Shibaura, Minato-ku, Tokyo 108-0023, Japan}

\address{$^2$ Graduate School of Information Sciences, Tohoku University, Sendai 980-8579, Japan}

\address{$^3$ RIKEN, Interdisciplinary Theoretical and Mathematical Sciences (iTHEMS), Wako, Saitama 351-0198, Japan}

\begin{abstract}
We derive a generic bound on the rate of decrease of transverse field for quantum annealing to converge to the ground state of a generic Ising model when quantum annealing is formulated as an infinite-time process.  Our theorem is based on a rigorous upper bound on the excitation probability in the infinite-time limit and is a mathematically rigorous counterpart of a previously known result derived only from the leading-order term of the asymptotic expansion of adiabatic condition. Since our theorem gives a sufficient condition of convergence for a generic transverse-field Ising model, any specific problem may allow a better, faster, control of the coefficient.
\end{abstract}

\section{Introduction}

Quantum annealing is an active field of research for combinatorial optimization, sampling, and quantum simulation \cite{Kadowaki1998,Albash2018,Hauke2020}.  There nevertheless exist only a very limited number of studies on mathematical conditions for convergence to the correct solution, the ground state of an Ising model, for generic or specific problems.  One of such studies is the contribution by Morita and Nishimori \cite{Morita2007,Morita2008}, where a sufficient condition is discussed for the amplitude of transverse field to satisfy in the long-time limit for a generic problem. The result is a power-law (polynomial) decrease of the amplitude as a function of time, which is faster than the corresponding rate of temperature decrease in classical simulated annealing \cite{Geman1984}.  

The approach of Morita and Nishimori builds on an approximate version of the ``adiabatic theorem", which takes into account only the leading-order term of the asymptotic expansion of excitation probability and ignores higher order contributions.  The goal of the present paper is to fix this insufficiency and derive a mathematically rigorous condition for convergence in the infinite-time limit based on the rigorous adiabatic theorem by Jansen, Ruskai, and Seiler \cite{Jansen2007}.

We formulate the problem and present our result with its proof in the next section.  The final section concludes the paper with discussions.

\section{Convergence condition}

\subsection{Formulation}
\label{sec:formulation}

Let us consider the following time-dependent Hamiltonian,
\begin{equation}
\label{eq:total_Hamiltonian}
H(t)=H_{\rm Ising}+H_{\rm TF}(t),
\end{equation}
where $H_{\rm Ising}$ is an arbitrary Ising Hamiltonian with general many-body interactions,
\begin{equation}
\label{eq:Ising_Hamiltonian}
H_{\rm Ising} = -\sum_{i=1}^N J_i\sigma^z_i - \sum_{i,j} J_{ij}\sigma^z_i \sigma^z_j - \sum_{i,j,k} J_{ijk}\sigma^z_i \sigma^z_j \sigma^z_k - \cdots,
\end{equation}
and $H_{\rm TF}(t)$ stands for a transverse-field term with time-dependent coefficient,
\begin{equation}
\label{eq:TF}
H_{\rm TF}(t)=-\Gamma(t)\sum_{i=1}^N \sigma_i^x,
\end{equation}
where $N$ is the total number of spins and $\sigma_i^x$ and $\sigma_i^z$ denote the $x$ and $z$ components of the Pauli matrix at site $i$, respectively. The coefficients $\{J_i, J_{ij}, J_{ijk}\}$ in the Ising Hamiltonian are supposed to be scaled such that $H_{\rm Ising}$ is of $\mathcal{O}(N)$
\footnote{
This condition is not essential and can indeed be relaxed to an arbitrary polynomial of $N$ since leading order terms in the following discussion are of exponential order.
}.
Time $t$ is supposed to run from 0 to infinity. The coefficient $\Gamma (t)$ will decrease from its large initial value $\Gamma (0)$ to zero as $t\to \infty$. Our goal is to derive a sufficient condition for the function $\Gamma (t)$ to satisfy in order for the system to reach the ground state of the generic Ising model of equation (\ref{eq:Ising_Hamiltonian}) within a given precision in the infinite time limit $t\to\infty$. Notice that we do not discuss computational complexity of quantum annealing for the generic problem (\ref{eq:total_Hamiltonian}) with (\ref{eq:Ising_Hamiltonian}), i.e., the amount of (finite but large) time to reach the solution as a function of the problem size $N$, which is already known to be NP hard \cite{Barahona1982}.

\subsection{Adiabatic theorem}
\label{sec:adiabatic_theorem}

Our theory relies on the adiabatic theorem proved by Jansen, Ruskai, and Seiler \cite{Jansen2007}.  Suppose that the Hamiltonian $H(t)$ depends on $t$ through a dimensionless time scaled by finite computation time $\tau$, $s=t/\tau$,
\begin{align}
    \tilde{H}(s)\equiv H(t),
\end{align}
where $s$ runs from 0 to 1 and thus $t$ runs from 0 to $\tau$. The Schr\"odinger equation reads
\begin{align}
    i\frac{1}{\tau}\frac{\partial}{\partial s}\ket{\psi_{\tau}(s)}=\tilde{H}(s)\ket{\psi_{\tau}(s)}.
\end{align}
The reduced Planck constant $\hbar$ is set to 1 for simplicity. We assume that $\tilde{H}(s)$ is twice differentiable by $s$ and the instantaneous ground state of $\tilde{H}(s)$ is non-degenerate. This latter condition is automatically satisfied by the transverse-field Ising model according to the Perron-Frobenius theorem \cite{Pillai2005} except in the limit $t\to\infty$, i.e., the pure Ising model, which may have ground-state degeneracy. We assume that this final degeneracy does not exist, which does not severly limit the type of problems.

Jansen {\em et al.} \cite{Jansen2007} proved the following inequality,
\begin{equation}
\label{eq:adiabatic_s}
\lVert P_\tau(s)-P(s)\rVert\le  \frac{\left\lVert\frac{d\tilde{H}(0)}{ds}\right\rVert}{\tau \Delta(0)^2}+\frac{\left\lVert\frac{d\tilde{H}(s)}{ds}\right\rVert}{\tau \Delta(s)^2} 
 +\frac{1}{\tau} \int_0^s d\tilde{s} \left(\frac{\left\lVert\frac{d^2\tilde{H}(\tilde{s})}{d\tilde{s}^2}\right\rVert}{\Delta(\tilde{s})^2} + \frac{7\, \left\lVert\frac{d\tilde{H}(\tilde{s})}{d\tilde{s}}\right\rVert^2}{\Delta(\tilde{s})^3}\right),
\end{equation}
where $\lVert \cdots \rVert$ denotes the operator norm.
Here $\Delta (s)$ stands for the instantaneous energy gap between the ground state and the first excite state, $\Delta (s)=\epsilon_1(s)-\epsilon_0(s)$, where $\epsilon_0(s)$ and $\epsilon_1(s)$ are instantaneous ground-state and first-excited-state energies of $\tilde{H}(s)$,
\begin{align}
    \tilde{H}(s)\ket{j(s)}=\epsilon_{j}(s)\ket{j(s)}~(j=0, 1, \cdots).
\end{align}
On the left-hand side of equation (\ref{eq:adiabatic_s}) appear projectors onto the current running state $\ket{\psi_{\tau}(s)}$ and the instantaneous ground state $\ket{0(s)}$, respectively,
\begin{align}
    P_{\tau}(s)=\ket{\psi_{\tau}(s)}\bra{\psi_{\tau}(s)},~
    P(s)=\ket{0(s)}\bra{0(s)}.
\end{align}
One can verify that the left-hand side of equation (\ref{eq:adiabatic_s}) bounds the square root of the probability of excitation,
\begin{align}
\label{eq:bound op sq}
    \lVert P_\tau(s)-P(s)\rVert\ge \sqrt{\sum_{j=1}^\infty |c_j(s)|^2},
\end{align}
where $c_j(s)$ is the coefficient of expansion of $\ket{\psi_{\tau}(s)}$ in terms of $\ket{j(s)}$,
\begin{align}
    \ket{\psi_{\tau}(s)}=\sum_{j=0}^{\infty} c_j(s)\ket{j(s)}.
\end{align}
This can be seen as follows: 
\begin{equation}
\begin{split}
    &\left(P_\tau(s)-P(s)\right)^2 =\left(\ket{\psi_{\tau}(s)}\bra{\psi_{\tau}(s)} - \ket{0(s)}\bra{0(s)}\right)^2 \\
    &= \ket{\psi_{\tau}(s)}\bra{\psi_{\tau}(s)} + \ket{0(s)}\bra{0(s)} -c_0(s)\ket{0(s)}\bra{\psi_{\tau}(s)}- \overline{c_0(s)}\ket{\psi_{\tau}(s)}\bra{0(s)}.
    \end{split}
\end{equation}
When this is acted on $\ket{0(s)}$, we have 
\begin{equation}
\begin{split}
    &\left(\ket{\psi_{\tau}(s)}\bra{\psi_{\tau}(s)} + \ket{0(s)}\bra{0(s)} -c_0(s)\ket{0(s)}\bra{\psi_{\tau}(s)}- \overline{c_0(s)}\ket{\psi_{\tau}(s)}\bra{0(s)} \right)\ket{0(s)}\\
    &=\overline{c_0(s)}\ket{\psi_{\tau}(s)}+\ket{0(s)}-|c_0(s)|^2\ket{0(s)}-\overline{c_0(s)}\ket{\psi_{\tau}(s)}\\
    &=\left(1-|c_0(s)|^2\right)\ket{0(s)}=\sum_{j=1}^\infty |c_j(s)|^2\ket{0(s)}.
\end{split}    
\end{equation}
This shows the bound \eqref{eq:bound op sq}.
Thus it is required to keep $\lVert P_\tau(s)-P(s)\rVert$ small if we demand that the system stays close to the instantaneous ground state at any $s$.

It is straightforward to rewrite equation (\ref{eq:adiabatic_s}) in terms of $t$, in place of $s$, as
\begin{equation}
\label{eq:adiabatic_t}
\lVert P_\tau(t)-P(t)\rVert\le  \frac{\left\lVert\frac{dH(0)}{dt}\right\rVert}{\Delta(0)^2}+\frac{\left\lVert\frac{dH(t)}{dt}\right\rVert}{\Delta(t)^2} 
 + \int_0^t d\tilde{t} \left(\frac{\left\lVert\frac{d^2 H(\tilde{t})}{d\tilde{t}^2}\right\rVert}{\Delta(\tilde{t})^2} + \frac{7\, \left\lVert\frac{dH(\tilde{t})}{d\tilde{t}}\right\rVert^2}{\Delta(\tilde{t})^3}\right).
\end{equation}
One can verify that the derivation process of equation (\ref{eq:adiabatic_s}) in reference \cite{Jansen2007} remains valid if we replace $s$ by $t(=s\tau)$ to reach equation \eqref{eq:adiabatic_t} without using $\tau$.

A commonly-used form of adiabatic condition 
\begin{align}
    \tau \gg \frac{\lvert \bra{1(s)}\frac{d\tilde{H}}{ds}\ket{0(s)}\rvert}{\Delta (s)^2}
\end{align}
corresponds to keeping small the value of only the second term of the right-hand side of equation (\ref{eq:adiabatic_s}).
We use full equation (\ref{eq:adiabatic_t}), one of rigorous versions of adiabatic theorem \cite{Albash2018}, to derive a sufficient condition for convergence of quantum annealing in the limit $t\to\infty$.

We are interested in suppressing the final probability of excitation, which is evaluated by taking the limit $t\to\infty$ on both sides of equation (\ref{eq:adiabatic_t}). 
\begin{align}
\label{eq:adiabatic_infinite}
    P_{\rm excited}\le  \frac{\left\lVert\frac{dH(0)}{dt}\right\rVert}{\Delta(0)^2}+\lim_{t \to \infty}\frac{\left\lVert\frac{dH(t)}{dt}\right\rVert}{\Delta(t)^2} 
 + \int_0^{\infty} dt \left(\frac{\left\lVert\frac{d^2 H(t)}{dt^2}\right\rVert}{\Delta(t)^2} + \frac{7\, \left\lVert\frac{dH(t)}{dt}\right\rVert^2}{\Delta(t)^3}\right),
\end{align}
where the left-hand side  
\begin{align}
    P_{\rm excited}=\lim_{t\to\infty}\lVert P_\tau(t)-P(t)\rVert
\end{align}
bounds the square root of the final excitation probability.
This is a weaker condition than requiring adiabaticity in the whole range of annealing process by imposing the adiabatic condition (\ref{eq:adiabatic_t}) for all $0<t<\infty$. 

\subsection{Evaluation of integrand}
\label{sec:evaluation_integrand}

We now aim to find a condition on the coefficient $\Gamma(t)$ to make the bound on the excitation probability $P_{\rm excited}$ on the right-hand side of the equation \eqref{eq:adiabatic_infinite} arbitrarily small.

To achieve this goal, first we evaluate the derivatives of the Hamiltonian. Since $H(t)$ depends on $t$ only through $\Gamma (t)$, we can easily obtain the following bounds,
\begin{equation}
\label{eq:first_deriv}
\left\lVert \frac{dH(t)}{dt}\right\rVert=\lvert\Gamma'(t)\rvert \left\lVert\sum_{i=1}^N \sigma_i^x\right\rVert= N\, |\Gamma'(t)|.
\end{equation}
Similarly,
\begin{equation}
\label{eq:second_deriv}
\left\lVert\frac{d^2 H(t)}{dt^2}\right\rVert= N\, \lvert\Gamma''(t)\rvert.
\end{equation}

The non-trivial part is to estimate a lower bound on the energy gap $\Delta (t)$, but this problem has already been solved for generic $H(t)$ in references \cite{Somma2007,Morita2007,Morita2008} as
\begin{align}
    \Delta (t)\ge A\,\Gamma (t)^N,
    \label{eq:bound_Delta}
\end{align}
where $A$ is independent of $t$ but depends on $N$ asymptotically $(N\gg 1)$ as
\begin{align}
    A= a \sqrt{N}e^{-bN} \label{eq:A}
\end{align}
with $N$-independent positive constants $a$ and $b$.
We find that replacement of denominators and numerators of two terms in the integrand of equation \eqref{eq:adiabatic_infinite} by equations \eqref{eq:first_deriv}, \eqref{eq:second_deriv},  \eqref{eq:bound_Delta}, and \eqref{eq:A} leads to their asymptotic (in $N$) upper bounds
as
\begin{align}
    \frac{\left\lVert\frac{d^2 H(t)}{dt^2}\right\rVert}{\Delta(t)^2} &\le\frac{e^{2bN}\lvert \Gamma''(t)\rvert}{a^2 \Gamma (t)^{2N}}
    \label{eq:bound_integrand1},\\
   \frac{\left\lVert\frac{dH(t)}{dt}\right\rVert^2}{\Delta(t)^3}&\le\frac{\sqrt{N}\, e^{3bN} (\Gamma'(t))^2}{a^3 \Gamma (t)^{3N}} 
   \,
    \label{eq:bound_integrand2}
\end{align}
respectively. Those upper bounds are to decrease faster than $t^{-1}$ as $t$ tends to infinity for the integral to converge in equation \eqref{eq:adiabatic_infinite}.

\subsection{Condition on the coefficient}
\label{sec:condition_Gamma}

To understand what functional form is allowed for $\Gamma(t)$ under the above-derived condition, we express $\Gamma (t)$ as
\begin{align}
     \Gamma (t)= (\delta t+c)^{-g(t)}
     \label{eq:Gamma_g}
\end{align}
with a twice-differentiable function $g(t)$, which should be strictly positive $g(t)>0$ because $\Gamma(t)$ is expected to tend toward 0 as $t\to \infty$. $\delta$ denotes a small parameter, and $c$ is a positive nonzero constant of order $\mathcal{O}(N^0)$.
We prove the following theorem.
\begin{theorem}
\label{theorem}
Excitation probability in the infinite-time limit $P_{\rm excited}$ can be made arbitrarily small for a large but fixed system size $N$
if the function $g(t)$ in equation \eqref{eq:Gamma_g} satisfies the following conditions,
\begin{align}
    &0<g(t)\le L
    \label{eq:g-condition1},\\
    &
    \lvert g'(t)\rvert \le \frac{\delta\, c'}{(\delta t+c)^{1+l}},
    \label{eq:g-condition2}\\
    & \lvert g''(t)\rvert \le \frac{\delta^2\, c'' \,(\delta t+c)^{-1-(2N-1)/(3N-2)}}{\log (\delta t+c)},
    \label{eq:g-condition3}
\end{align}
with a positive constant $L$ (which may depend on $N$) that satisfies the strict inequality $L<\frac{1}{3N-2}$, and positive constants $l$, $c'$ and $c''$, and $\delta$ is chosen small enough, of the order of a small constant multiplied by $N^{-1/2}e^{-3bN}$.
\end{theorem}

\begin{proof}
We demonstrate that each term on the right-hand side of equation \eqref{eq:adiabatic_infinite} can be made arbitrarily small, of the order of $\delta$ or smaller, under the conditions \eqref{eq:g-condition1}, \eqref{eq:g-condition2}, and \eqref{eq:g-condition3}. 

First, we evaluate the first two terms on the right-hand side of equation \eqref{eq:adiabatic_infinite}. From equations \eqref{eq:first_deriv}, \eqref{eq:bound_Delta}, and \eqref{eq:A}, we obtain the following bound:
\begin{equation}
\frac{\left\lVert\frac{d H(t)}{dt}\right\rVert}{\Delta(t)^2} \le\frac{e^{2bN}\lvert \Gamma'(t)\rvert}{a^2 \Gamma (t)^{2N}}.
\end{equation}
We have 
\begin{equation}
\label{eq:Gamma_prime}
\Gamma'(t)=\frac{d}{dt}(\delta t+c)^{-g(t)}=(\delta t+c)^{-g(t)} \left(-g'(t)\log (\delta t+c)-\frac{\delta \,g(t)}{\delta t+c} \right).
\end{equation}
Owing to the condition \eqref{eq:g-condition2}, 
\begin{equation}
\lvert -g'(t)\log (\delta t+c)\rvert\le \frac{\delta\, c'}{\delta t+c}\frac{\lvert\log (\delta t+c)\rvert}{(\delta t+c)^l}.
\end{equation}
Because $\lim_{t\to\infty}\frac{\log (\delta t+c)}{(\delta t+c)^l}=0$, $\frac{\lvert\log (\delta t+c)\rvert}{(\delta t+c)^l}$ is bounded from above by a finite number over the region $0\le t<\infty$. This means that the function $\frac{\lvert\log (\delta t+c)\rvert}{(\delta t+c)^l}$ has the finite maximum over the region $0\le t<\infty$. We denote this maximum by $m$, i.e. 
\begin{equation}
m:={\rm max}_{0\le t<\infty} \frac{\lvert\log (\delta t+c)\rvert}{(\delta t+c)^l}.
\end{equation}
From these and using the condition on $g(t)$ \eqref{eq:g-condition1}, we obtain the following bound on $\lvert \Gamma'(t)\rvert$:
\begin{eqnarray}
\label{eq:bound on Gamma}
\lvert \Gamma'(t)\rvert &&\le \delta (\delta t+c)^{-g(t)-1}(g(t)+mc') \\ \nonumber
&&\le \delta (\delta t+c)^{-g(t)-1}(L+mc').
\end{eqnarray}
Then, we obtain the relation 
\begin{equation}
\label{eq:first two terms}
\frac{\left\lVert\frac{d H(t)}{dt}\right\rVert}{\Delta(t)^2} \le\frac{e^{2bN}\lvert \Gamma'(t)\rvert}{a^2 \Gamma (t)^{2N}}\le \frac{e^{2bN}}{a^2}\delta (\delta t+c)^{(2N-1)g(t)-1}(L+mc').
\end{equation}
Owing to equation \eqref{eq:g-condition1}, $(2N-1)g(t)-1<0$; therefore, from the relation \eqref{eq:first two terms} we immediately conclude that $\lim_{t \to \infty}\frac{\left\lVert\frac{dH(t)}{dt}\right\rVert}{\Delta(t)^2}=0$.

We also obtain from the relation \eqref{eq:first two terms} the evaluation of the term $\frac{\left\lVert\frac{dH(0)}{dt}\right\rVert}{\Delta(0)^2}$ as
\begin{equation}
\label{eq:eval H0}
\frac{\left\lVert\frac{dH(0)}{dt}\right\rVert}{\Delta(0)^2}\le \frac{e^{2bN}}{a^2}\delta c^{(2N-1)g(0)-1}(L+mc').
\end{equation} 

Now we evaluate the second term in the integral in equation \eqref{eq:adiabatic_infinite}. Utilizing the bounds \eqref{eq:bound_integrand2} and \eqref{eq:bound on Gamma}, we obtain the following relation:
\begin{eqnarray}
\int_0^{\infty} dt \frac{\left\lVert\frac{dH(t)}{dt}\right\rVert^2}{\Delta(t)^3} &&\le \frac{\sqrt{N}\, e^{3bN}}{a^3}\int_0^{\infty} dt\frac{(\Gamma'(t))^2}{\Gamma (t)^{3N}}\\ \nonumber
&& \le \frac{\sqrt{N}\, e^{3bN}}{a^3}(L+mc')^2 \int_0^{\infty} dt \delta^2 (\delta t+c)^{(3N-2)g(t)-2}.
\end{eqnarray}
We split the integral into two parts, one for $\delta\, t<1$ and the other for $\delta\, t>1$ and show that each of them is of the order of $\delta$.
Since $0<g(t)\le L<\frac{1}{3N-2}$, 
\begin{equation}
-2<(3N-2)g(t)-2\le (3N-2)L-2<-1;
\end{equation}
therefore, we have $(\delta t+c)^{(3N-2)g(t)-2}\le (\delta t+c)^{(3N-2)L-2}$ for $1/\delta\le t<\infty$. Thus,
\begin{equation}
\begin{split}
\label{eq:comp_1}
&\int_0^{\infty} dt \delta^2 (\delta t+c)^{(3N-2)g(t)-2} \\
&=\int_0^{\frac{1}{\delta}} dt \delta^2 (\delta t+c)^{(3N-2)g(t)-2}+\int_{\frac{1}{\delta}}^{\infty} dt \delta^2 (\delta t+c)^{(3N-2)g(t)-2} \\
& \le \int_0^{\frac{1}{\delta}} dt \delta^2 c^{(3N-2)g(t)-2}+\int_{\frac{1}{\delta}}^{\infty} dt \delta^2 (\delta t+c)^{(3N-2)L-2}.
\end{split}
\end{equation}
In the last line in equation \eqref{eq:comp_1}, we used the fact that $(\delta t+c)^{(3N-2)g(t)-2}\le c^{(3N-2)g(t)-2}$ because $(3N-2)g(t)-2<0$ and $\delta t\ge 0$. For the first integral in the last line in equation \eqref{eq:comp_1}, utilizing a change of variable $u=\delta t$ with a function $\tilde{g}(u):=g(t)$, 
\begin{equation}
\label{eq:suppl_1}
\int_0^{\frac{1}{\delta}} dt \delta^2 c^{(3N-2)g(t)-2}=\delta \int_0^1 du\, c^{(3N-2)\tilde{g}(u)-2}< \delta\, {\rm max}\{c^{-1}, c^{-2}\}
\end{equation}
since $-2< (3N-2)g(u)-2 <-1$.
For the second integral in the last line in \eqref{eq:comp_1}, because $(3N-2)L-1<0$, we have
\begin{eqnarray}
\label{eq:suppl_2}
&&\int_{\frac{1}{\delta}}^{\infty} dt \delta^2 (\delta t+c)^{(3N-2)L-2} \\ \nonumber
&&=\frac{\delta}{(3N-2)L-1}\left[(\delta t+c)^{(3N-2)L-1} \right]_\frac{1}{\delta}^\infty =\frac{\delta}{1-(3N-2)L}(c+1)^{(3N-2)L-1}.
\end{eqnarray}
Thus, we learn that 
\begin{equation}
\label{eq:eval_bound_1}
\int_0^{\infty} dt \frac{\left\lVert\frac{dH(t)}{dt}\right\rVert^2}{\Delta(t)^3}< \delta\frac{\sqrt{N}\, e^{3bN}}{a^3}(L+mc')^2 \left(\frac{(c+1)^{(3N-2)L-1}}{1-(3N-2)L}+{\rm max}\{c^{-1}, c^{-2}\}\right).
\end{equation}

Finally, we evaluate the first term in the integral in the equation \eqref{eq:adiabatic_infinite}. From equation \eqref{eq:Gamma_prime},
\begin{eqnarray}
\Gamma''(t)&&=\frac{d^2}{dt^2}(\delta t+c)^{-g(t)}\\ \nonumber
&&=(\delta t+c)^{-g(t)} \left(\frac{\delta^2 g(t)}{(\delta t+c)^2}-g''(t)\log(\delta t+c)-\frac{2\delta g'(t)}{\delta t+c}\right)\\ \nonumber
&&+ (\delta t+c)^{-g(t)} \left(-g'(t)\log (\delta t+c)-\frac{\delta g(t)}{\delta t+c} \right)^2.
\end{eqnarray}
Thus, utilizing the conditions \eqref{eq:g-condition1}, \eqref{eq:g-condition2}, and \eqref{eq:g-condition3}, we obtain
\begin{eqnarray}
\lvert \Gamma''(t)\rvert &&\le \delta^2 (\delta t+c)^{-g(t)-2}\left(g(t)+\frac{2c'}{(\delta t+c)^l}\right)+c''\delta^2 (\delta t+c)^{-g(t)-1-\frac{2N-1}{3N-2}}\\ \nonumber
&&+(\delta t+c)^{-g(t)} \delta^2 \left(\frac{g(t)+mc'}{\delta t+c}\right)^2 \\ \nonumber
&&\le \delta^2 (\delta t+c)^{-g(t)-2}\left(L+\frac{2c'}{c^l}+(L+mc')^2\right)+c''\delta^2 (\delta t+c)^{-g(t)-1-\frac{2N-1}{3N-2}}.
\end{eqnarray}
This yields the bound as follows:
\begin{equation}
\frac{\lvert \Gamma''(t)\rvert}{\Gamma (t)^{2N}}\le \delta^2 (\delta t+c)^{(2N-1)g(t)-2}\left(L+\frac{2c'}{c^l}+(L+mc')^2\right)+c''\delta^2 (\delta t+c)^{(2N-1)g(t)-1-\frac{2N-1}{3N-2}}.
\end{equation}
From this and using \eqref{eq:bound_integrand1}, we learn that 
\begin{eqnarray}
\int_0^{\infty} dt \frac{\left\lVert\frac{d^2 H(t)}{dt^2}\right\rVert}{\Delta(t)^2}&&\le \frac{e^{2bN}}{a^2}\int_0^{\infty} dt\frac{\lvert \Gamma''(t)\rvert}{\Gamma (t)^{2N}}\\ \nonumber
&&\le \frac{e^{2bN}}{a^2} \left(L+\frac{2c'}{c^l}+(L+mc')^2\right) \int_0^{\infty} dt \delta^2 (\delta t+c)^{(2N-1)g(t)-2} \\ \nonumber
&&+ \frac{e^{2bN}}{a^2} c''\int_0^{\infty} dt \delta^2 (\delta t+c)^{(2N-1)g(t)-1-\frac{2N-1}{3N-2}}.
\end{eqnarray}
After a computation analogous to that we performed in equations \eqref{eq:comp_1}, \eqref{eq:suppl_1}, and \eqref{eq:suppl_2} replacing $3N-2$ with $2N-1$, we find
\begin{equation}
\int_0^{\infty} dt \delta^2 (\delta t+c)^{(2N-1)g(t)-2}< \delta\left(\frac{(c+1)^{(2N-1)L-1}}{1-(2N-1)L}+{\rm max}\{c^{\frac{2N-1}{3N-2}-2}, c^{-2}\}\right),
\end{equation}
and 
\begin{equation}
\int_0^{\infty} dt \delta^2 (\delta t+c)^{(2N-1)g(t)-1-\frac{2N-1}{3N-2}}< \delta\left(\frac{(c+1)^{(2N-1)L-\frac{2N-1}{3N-2}}}{\frac{2N-1}{3N-2}-(2N-1)L}+{\rm max}\{c^{-1}, c^{-1-\frac{2N-1}{3N-2}}\}\right).
\end{equation}
From these, we deduce the following bound:
\begin{eqnarray}
\label{eq:eval_bound_2}
&&\int_0^{\infty} dt \frac{\left\lVert\frac{d^2 H(t)}{dt^2}\right\rVert}{\Delta(t)^2}\le \frac{e^{2bN}}{a^2}\int_0^{\infty} dt\frac{\lvert \Gamma''(t)\rvert}{\Gamma (t)^{2N}}\\ \nonumber
&&< \delta\frac{e^{2bN}}{a^2} \left(L+\frac{2c'}{c^l}+(L+mc')^2\right) \left(\frac{(c+1)^{(2N-1)L-1}}{1-(2N-1)L}+{\rm max}\{c^{\frac{2N-1}{3N-2}-2}, c^{-2}\}\right) \\ \nonumber
&&+ \delta\frac{e^{2bN}}{a^2} c''\left(\frac{(c+1)^{(2N-1)L-\frac{2N-1}{3N-2}}}{\frac{2N-1}{3N-2}-(2N-1)L}+{\rm max}\{c^{-1}, c^{-1-\frac{2N-1}{3N-2}}\}\right).
\end{eqnarray}

Thus, from evaluation of bounds on terms \eqref{eq:eval H0}, \eqref{eq:eval_bound_1}, and \eqref{eq:eval_bound_2} together with that $\lim_{t \to \infty}\frac{\left\lVert\frac{dH(t)}{dt}\right\rVert}{\Delta(t)^2}=0$, we conclude that the right-hand side of equation \eqref{eq:adiabatic_infinite} can be made arbitrarily small for fixed $N$ by choosing a sufficiently small $\delta$ of the order of a small constant multiplied by $N^{-1/2}e^{-3bN}$ as one sees in equations \eqref{eq:eval H0}, \eqref{eq:eval_bound_1}, and \eqref{eq:eval_bound_2}. This shows that the excitation probability can be made arbitrarily small under the stated conditions.
\end{proof}
\noindent
{\em Example.}~ As a simple example, one may choose a constant function,
\begin{align}
    g(t)=\frac{1}{4N}.
\end{align}

\subsection{Bounded coefficient}
\label{sec:condition_s}

It is often the case that one considers the following form of the 
Schr\"odinger dynamics:
\begin{equation}
\label{standard_eq1}
i\frac{d\psi(t)}{dt}=\left(s(t) H_{\rm Ising}-\left(1-s(t) \right) \sum_i \sigma_i^x\right) \psi(t),
\end{equation}
with a monotonically increasing function $0\le s(t)\le 1$,
instead of equation \eqref{eq:total_Hamiltonian} with equation \eqref{eq:TF}. We show that it is possible to rewrite the theory developed in previous sections to this case if we choose $t$ to run from 0 to $\infty$
\footnote{
Notice that $t$ in the formulation of equation \eqref{standard_eq1} is often supposed to run within a finite interval $0\le t\le \tau$. Our theory does not apply to this case of finite-time development. 
}.

Equation \eqref{standard_eq1} can be rewritten as
\begin{equation}
\label{standard_eq2}
i\frac{1}{s(t)}\frac{d\psi(t)}{dt}=\left( H_{\rm Ising}-\frac{1-s(t)}{s(t)} \sum_i \sigma_i^x \right) \psi(t).
\end{equation}
Let us define 
\footnote{
Not to be confused with the running variable in the integral in equation \eqref{eq:adiabatic_t}.
}
\begin{equation}
\label{tilde t}
\tilde{t}:=\int_0^t s(t) dt,
\end{equation}
which is a monotonic function of $t$. Utilizing $\tilde{t}$, equation \eqref{standard_eq2} can be rewritten as 
\begin{equation}
\label{standard eq3}
i\frac{d\psi(t)}{d\tilde{t}}=\left(H_{\rm Ising}-\frac{1-s(t)}{s(t)} \sum_i \sigma_i^x\right) \psi(t).
\end{equation}
With a function, $\Gamma(\tilde{t}):=(1-s(t))/s(t)$, equation \eqref{standard eq3} becomes 
\begin{equation}
\label{standard eq4}
i\frac{d\psi(t)}{d\tilde{t}}=\left( H_{\rm Ising}-\Gamma(\tilde{t}) \sum_i\sigma_i^x \right) \psi(t).
\end{equation}
Application of the argument in previous sections to \eqref{standard eq4} shows that, for convergence as $\tilde{t} \to \infty$, it is sufficient that $\Gamma(\tilde{t})$ behaves as
\begin{equation}
\label{eq for s}
\Gamma(\tilde{t})=\frac{1-s(t)}{s(t)} \propto \tilde{t}^{-\tilde{g}(\tilde{t})},
\end{equation}
with $\tilde{g}(\tilde{t})$ satisfying equations \eqref{eq:g-condition1} to \eqref{eq:g-condition3}.
Solving \eqref{eq for s} for $s(t)$, one obtains
\begin{equation}
\label{s asymptotic}
s(t)=\frac{1}{1+\tilde{t}^{-\tilde{g}(\tilde{t})}} \approx 1-\tilde{t}^{-\tilde{g}(\tilde{t})}~ (\tilde{t}\gg 1).
\end{equation}

We remark that in some cases $\tilde{t}$ can be approximated by $t$ when $t$ is large. For example, if we choose $s(t)={\rm tanh}\, t$,
\begin{equation}
\tilde{t}  = \log \cosh t\approx t ~ (t\gg1).
\end{equation}

\section{Discussion}
\label{sec:conclusion}
We have studied a sufficient condition for quantum annealing to converge to the ground state of a generic Ising model in the infinite-time limit for a given finite system size.  This is a mathematically rigorous version of a previous result \cite{Morita2007,Morita2008}, in which an approximate adiabatic condition was used. The result shows that convergence is achieved if the coefficient of the transverse-field term decreases with a power law of time or slower. This is qualitatively similar to the previous result in references \cite{Morita2007,Morita2008} but is different in rigorous quantification.  In particular, constraints on derivatives of the coefficient did not exist before.

Our result is a sufficient condition for a generic Ising model:  For any problem represented by equations \eqref{eq:total_Hamiltonian}, \eqref{eq:Ising_Hamiltonian}, and \eqref{eq:TF}, the system will become close to the ground state of the Ising model in the infinite-time limit if the conditions in Theorem 1 are satisfied.  We are unable to predict what happens if the conditions are not satisfied. It may happen that a faster decrease of $\Gamma (t)$ than we have derived here  results in convergence to the ground state for a given specific problem, or it may also be the case for other examples that the system ends up in an excited state even if one spends an infinite amount of time when the conditions in Theorem 1 are not met.

Our conclusion is to be contrasted with the corresponding result for classical simulated annealing \cite{Geman1984}. In a classical problem, the external parameter, temperature, is to be decreased as an inverse-logarithmic function of time in the limit of large computation time, which is much slower than the power law in the present quantum case.  However, we should be careful not to conclude that quantum annealing is more efficient than simulated annealing since in both cases one is supposed to spend an infinitely long time to reach the solution.

We discussed the behavior of the coefficient of the transverse-field term in the infinite-time limit.  No constraint is imposed on the behavior of the coefficient in the intermediate time region. 

It is important to remember that our goal is not to discuss computational complexity of quantum annealing.  Indeed, we have discussed a very generic Ising model, which is known to be NP hard \cite{Barahona1982}.  

We hope that developments along the line of the present work will lead to further non-trivial results to lay a firm theoretical foundation of quantum annealing.

\ack
We thank Kazuya Kaneko for valuable comments.
This work is based on a project JPNP16007 commissioned by the New Energy and Industrial Technology Development Organization (NEDO).  
\vspace{5mm}

\begin{thebibliography}{10}

\bibitem{Kadowaki1998}
T.~Kadowaki and H.~Nishimori.
\newblock {Quantum annealing in the transverse Ising model}.
\newblock {\em Phys. Rev. E}, 58:5355, 1998.

\bibitem{Albash2018}
T.~Albash and D.~A. Lidar.
\newblock {Adiabatic quantum computation}.
\newblock {\em Rev. Mod. Phys.}, 90:015002, 2018.

\bibitem{Hauke2020}
P.~Hauke, H.~G. Katzgraber, W.~Lechner, H.~Nishimori, and W.~D Oliver.
\newblock {Perspectives of quantum annealing: Methods and implementations}.
\newblock {\em Rep. Prog. Phys.}, 83:054401, 2019.

\bibitem{Morita2007}
S.~Morita and H.~Nishimori.
\newblock {Convergence of Quantum Annealing with Real-Time Schr{\"{o}}dinger
  Dynamics}.
\newblock {\em J. Phys. Soc. Jpn.}, 76:064002, 2007.

\bibitem{Morita2008}
S.~Morita and H.~Nishimori.
\newblock {Mathematical foundation of quantum annealing}.
\newblock {\em J, Math. Phys.}, 49:125210, 2008.

\bibitem{Geman1984}
S.~Geman and D.~Geman.
\newblock {Stochastic Relaxation, Gibbs Distributions, and the Bayesian
  Restoration of Images}.
\newblock {\em IEEE Trans. Patt. Analy. Mach. Intel.}, PAMI-6:721--741, 1984.

\bibitem{Jansen2007}
S.~Jansen, M.-Beth. Ruskai, and R.~Seiler.
\newblock {Bounds for the adiabatic approximation with applications to quantum
  computation}.
\newblock {\em J. Math. Phys.}, 48:102111, 2007.

\bibitem{Barahona1982}
F~Barahona.
\newblock {On the computational complexity of Ising spin glass models}.
\newblock {\em J. Phys. A}, 15:3241, 1982.

\bibitem{Pillai2005}
T.~Suel S.~U.~Pillai and S.~Cha.
\newblock {The Perron-Frobenius theorem: some of its application}.
\newblock {\em IEEE Sig. Proc.}, 22:62, 2005.

\bibitem{Somma2007}
R.~Somma, C.~Batista, and G.~Ortiz.
\newblock {Quantum Approach to Classical Statistical Mechanics}.
\newblock {\em Phys. Rev. Lett.}, 99(3):030603, jul 2007.

\end{thebibliography}

\end{document}